\documentclass[journal]{IEEEtran}
\usepackage{amsmath,amssymb,amsthm,mathtools}
\usepackage{microtype}
\usepackage{subcaption}
\usepackage{placeins}
\usepackage{stfloats}
\usepackage{float}
\usepackage{cite}
\usepackage[colorlinks=false]{hyperref}
\newtheorem{lemma}{Lemma}
\newtheorem{theorem}{Theorem}

\newtheorem{proposition}{Proposition}

\begin{document}
\title{Density control of large-scale multi-agent systems via tuning of agent interactions}

\author{%
Mena~Zkhary
and~Alain~Boldini%
\thanks{M. Zkhary and A. Boldini are with the Department of Mechanical Engineering, New York Institute of Technology, Old Westbury, NY 11568 USA.}
\thanks{Corresponding author: Alain Boldini (email: aboldini@nyit.edu).}
}
\maketitle
\begin{abstract}
Collective behavior in biological systems emerges through adaptive interactions among units, enabling large populations to coordinate without centralized control. This mechanism raises the question of whether collective behavior in engineered multi-agent systems can likewise be regulated by modifying the interaction law between agents, rather than by directly controlling individual units. To answer this question, we consider large-scale multi-agent systems evolving on a ring and treat the interaction kernel as a time-varying control variable whose evolution regulates the collective dynamics. Using a mean-field formulation, we develop a framework for the analysis and design of kernel-based controllers. We derive explicit conditions for controllability in application-relevant scenarios, and synthesize a state-feedback controller to regulate the kernel dynamics, validated on the microscopic, discrete dynamics. Numerical simulations demonstrate the approach's effectiveness and robustness to unmodeled perturbations. This work constitutes a first step in the design of bio-inspired control strategies for large-scale multi-agent systems by modulation of interactions between agents.
\end{abstract}
\section{Introduction}
%
%
%
%
Large-scale multi-agent systems (MASs) have emerged as a foundational paradigm to describe and control the collective behavior of large populations of autonomous agents in a wide range of engineering and scientific applications, including swarm robotics, intelligent transportation systems, biological collective behavior, and cyber--physical systems \cite{Brambilla2013,Ren2008,Murray2007,Vicsek2012,Rajkumar2010}. Through local interactions among individual agents, these systems exhibit complex collective behaviors that arise without centralized coordination. Understanding, predicting, and regulating such emergent behaviors have become fundamental challenges in control theory, particularly as the scale and complexity of MASs continue to grow.
%
%
%

Most models of populations of interacting agents are microscopic, as they describe the evolution of the dynamics of each agent through a set of coupled ordinary differential equations (ODEs). As the number of agents increases, these models become increasingly difficult to analyze, simulate, and use for control. Analogous to kinetic theory, which derives macroscopic fluid equations from the statistics of microscopic particle interactions, continuification establishes a systematic bridge between microscopic agent models and continuum mean-field descriptions of large populations. Through continuification, collective behavior is described in terms of the macroscopic density of the agents, along with other continuum fields. The evolution of these fields is governed by partial differential equations (PDEs), yielding scalable mathematical models that accurately capture the emergent macroscopic dynamics in the large-population limit \cite{Topaz2006,Carrillo2010,Bernoff,Leverentz,Nikitin}.
%
%
%

Building upon these mean-field models, substantial progress was made in recent years in the analysis and control of large-scale MASs. Mean-field games and mean-field optimal control established rigorous mathematical frameworks for distributed optimization, collective decision-making, and control of continuum interacting systems \cite{Lasry2007,Bensoussan2013,Carmona2018,Fornasier2014,Albi2017}. More recently, continuification-based approaches exploited mean-field representations to develop scalable frameworks for controller synthesis and density regulation in large-scale MASs \cite{Maffettone,Nikitin,Boldini,LeaderFollower,Plasticity,block2026continuification,giardina2026robotectonics}. Complementing these theoretical developments, advances in numerical methods for mean-field control significantly improved the computational tractability of large-scale interacting systems \cite{Hambly,Lauriere2021,Nurbekyan,Reisinger}. Collectively, these developments greatly expanded the analytical and computational tools available for understanding and regulating collective behavior using mean-field descriptions.
%
%
%

Existing control strategies regulate collective behavior through external control actions on individual agents, while treating the interaction law governing the collective dynamics
as fixed\cite{Maffettone,Boldini,LeaderFollower,Plasticity}. However, in many biological and human collective systems, coordinated behaviors emerge from adaptive local interactions rather than externally imposed control actions. For example, during predator avoidance, only a small subset of preys directly perceives the threat and responds to it by modulating its interactions with the predator and with other preys, such that coordinated escape responses rapidly propagate throughout the swarm without centralized coordination~\cite{Mateo2017,BouffanaisBook,Couzin2005}. Similarly, pedestrians continuously adapt their interactions with nearby individuals in response to their visual perception and the surrounding environment, giving rise to collective navigation and coordinated crowd motion \cite{Warren2022,Pan2025,Colas2022}. These observations highlight the fundamental role of modulation of local interactions in shaping collective behavior and motivate investigating whether systematically tuning the interaction law between agents can provide an alternative mechanism for regulating large-scale MASs.

The interaction modulation paradigm remains largely unexplored in mean-field control frameworks, even though related concepts have
been previously explored. In particular, we highlight two main lines of work that share similarities with interaction modulation.
%
%
%
 Controlled-Lagrangian and potential-shaping methods modify the system energy landscape to achieve desired closed-loop behaviors in mechanical systems, but are difficult to apply in large-scale MASs~\cite{Bloch2000,Bloch2001}. On the other hand, in mean-field synchronization models, optimal control strategies can tune the strength of prescribed  interaction terms, but do not modify the form of interactions between agents \cite{Sinigaglia2022}. In contrast to these existing approaches, the proposed interaction modulation framework treats the interaction kernel itself as the control variable, allowing its functional form to be systematically modified in order to regulate the collective dynamics.
%
%

To demonstrate the potential of interaction modulation, in this paper, we develop a Fourier-mode framework for regulating the collective dynamics of large-scale MASs evolving on a ring. Upon continuification of the microscopic model, the framework exploits physically relevant approximations to build a finite-dimensional description of the system dynamics, explicitly relating interaction modulation to the evolution of the collective motion. The proposed framework treats the interaction law as a controllable design variable, providing a mathematically tractable foundation for controller synthesis based on modulation of the kernel of interactions (Fig. \ref{fig:framework}). 
Upon designing the kernel modulation at the macroscopic level, modified interactions are implemented on individual agents through discretization.
To demonstrate the viability of the proposed framework, we investigate the controllability of the resulting interaction-modulated system by deriving explicit analytical conditions under which kernel modulation can regulate the dynamics near prescribed equilibrium configurations. These results provide theoretical insight into the capability and limitations of interaction-based control.
\begin{figure}[!tb]
    \centering
    \includegraphics[width=\columnwidth]{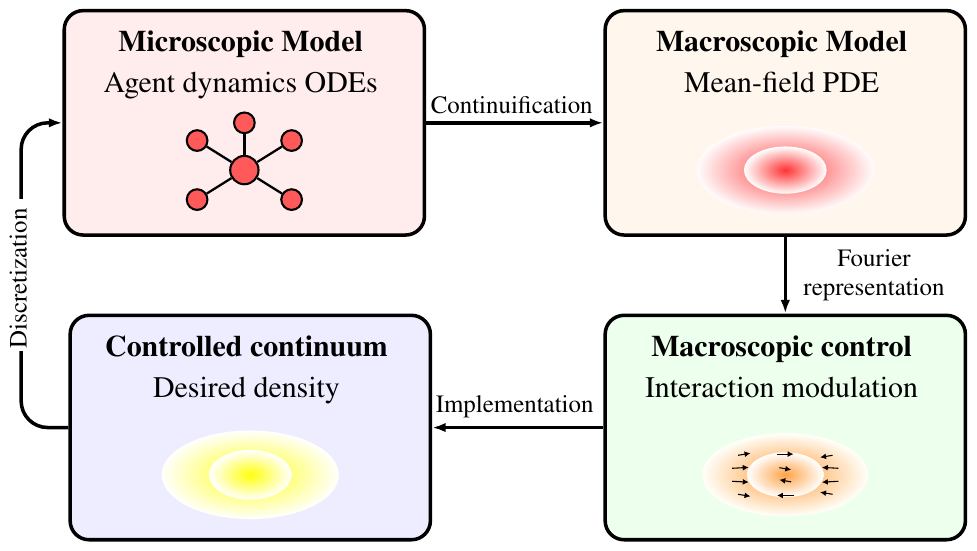}
    \caption{Overview of the proposed micro--macro interaction-based control framework. The microscopic interacting-agent system is transformed into a continuum mean-field model via continuification. At the macroscopic level, Fourier analysis is employed to modify the interaction kernel. The resulting interaction law is discretized to regulate the microscopic collective dynamics.}
    \label{fig:framework}
\end{figure}

The main contributions of this paper are summarized as follows.

\begin{enumerate}

\item We introduce an interaction-based control framework for mean-field models of large-scale MASs. Rather than relying on conventional external control inputs, the proposed framework regulates collective behavior through modulation of the interaction kernel.

\item We establish conditions for the kernel to make a desired density distribution an equilibrium. Next, we develop a state-feedback controller to regulate the linearized system dynamics near such an equilibrium, wherein modifications of the kernel serve as control inputs. We characterize the performance and robustness of such a controller under varying conditions.

\item We investigate the controllability properties of the linearized finite-dimensional approximation of the system dynamics in physically relevant conditions, with kernel modifications as control inputs. We establish explicit conditions under which the controllability of the dynamics around prescribed equilibrium densities can be guaranteed.

\end{enumerate}

The remainder of this paper is organized as follows. Section~\ref{sec:II} introduces the mathematical framework for interaction-based control and develops the state-feedback controller based on kernel modulation. Section~\ref{sec:controllability} investigates the controllability properties of the system controlled via interaction modulation. Section~\ref{numerical_results} demonstrates the proposed framework through numerical studies on both continuum and microscopic models. Section~\ref{sec:results} summarizes the main results of the paper and discusses future research directions.
\section{Mathematical Framework and Problem Formulation}
\label{sec:II}
\subsection{Microscopic and mean-field models}

We consider a swarm of \(N\) interacting agents moving on the one-dimensional unit circle
\(
\mathbb{S}^{1}
\),
parametrized by the angular position \(x\in[-\pi,\pi)\). We denote the position of agent \(i\) at time \(t\ge 0\) as
\(x_i(t)\in\mathbb{S}^{1}\). Following the widely adopted first-order kinematic model for collective motion \cite{Maffettone,Topaz2006,Carrillo2010}, the microscopic dynamics is described by a set of coupled ODEs,
\begin{equation}
\dot{x}_i(t)
=
\sum_{j=1}^{N}
K_{\mathrm{np}}\!\left([x_i(t),x_j(t)]_{\pi},t\right),
\qquad
i=1,\ldots,N,
\label{eq:micro}
\end{equation}
where
$
[x_i,x_j]_{\pi}
=
(x_i-x_j+\pi)\bmod(2\pi)-\pi$
denotes the wrapped angular difference between agents' positions, mapping the relative angular position to the interval
\([-\pi,\pi)\), and
\(K_{\mathrm{np}}:\mathbb{R}\times\mathbb{R}_{\ge0}\rightarrow\mathbb{R}\)
is a kernel describing pairwise interactions among the agents. Consistently with our approach, we consider the possibility of modulating the kernel by making it an explicit function of time.

Designing changes to the kernel in the microscopic dynamics is challenging for large-scale swarms. We thus adopt the corresponding continuified formulation, in which the swarm is described by its density
\(
\rho(x,t):\mathbb{S}^1\times \mathbb{R}_{\ge 0} \rightarrow \mathbb{R}_{\ge 0}
\). 
The evolution of the density is governed by the continuity equation
\begin{equation}
\frac{\partial\rho(x,t)}{\partial t}
+
\frac{\partial}{\partial x}
\left[
\rho(x,t)v(x,t)
\right]
=
0.
\label{eq:pde}
\end{equation}
Here, the velocity field $v:\mathbb{S}^{1}\times \mathbb{R}_0^+\rightarrow \mathbb{R}$ is found from the convolution 
\begin{equation}
v(x,t)
=
(K_\mathrm{np}*\rho)(x,t)
=
\int_{-\pi}^{\pi}
K(x-y,t)\rho(y,t)\,\mathrm{d}y,
\label{eq:velocity}
\end{equation}
where ``$*$'' denotes the convolution operator on
\(\mathbb{R}\), and $K\in L^2(\mathbb{S}^1)$ is the periodicized version of the non-periodic kernel $K_\mathrm{np}$ \cite{Boldini}. This operation is equivalent to the circular convolution of the periodicized kernel $ K$ with the density $\rho$ on $\mathbb{S}^1$. Both the density and the kernel are assumed to satisfy the periodicity conditions 
$
\rho(-\pi,t)
=
\rho(\pi,t),
K(-\pi,t)
=
K(\pi,t).
$
Without lack of generality, the density can be rescaled to satisfy a unitary mass normalization condition,
\begin{equation}
\int_{-\pi}^{\pi}
\rho(x,t)\,\mathrm{d}x
=
1.
\label{eq:mass}
\end{equation}
\subsection{Problem formulation}
The objective of this paper is to regulate the collective dynamics of the mean-field system by appropriately modifying the kernel $K(x,t)$, rather than by introducing external control inputs on the agents. Specifically, given a prescribed equilibrium density $\rho^{\mathrm d}(x)$ on the periodic domain $\mathbb{S}^{1}$, we seek to find a kernel $K(x,t)$ such that the density of the MASs asymptotically converges to the desired density almost everywhere, that is,
\begin{equation}
\lim_{t\rightarrow\infty}
\left\|
\rho(\cdot,t)-\rho^{\mathrm{d}}(\cdot)
\right\|_{L^2}
=0,
\label{eq:error_norm}
\end{equation}
provided that the initial density $\rho(x,0) = \rho_0(x)$ is close enough to $\rho^{\mathrm{d}}(x)$. Here, $||\cdot||_{L^2}$ represents the $L^2$ norm on $\mathbb{S}^1$. This problem resembles previous efforts in the literature~\cite{Maffettone}, but leverages the kernel $K(x,t)$ rather than an external input $u(x,t)$ as the control parameter in Eqs.~\eqref{eq:pde}--\eqref{eq:velocity}.
\subsection{Fourier representation}
Due to periodicity, the density field and the kernel admit a Fourier-series representation, 
\begin{subequations}\label{eq:fourier_expansions}
\begin{align}
\rho(x,t)
&=
\sum_{m=-\infty}^{\infty}
\hat{\rho}_m(t)e^{\mathrm{i}mx},
\label{eq:rho_complex}
\\
K(x,t)
&=
\sum_{m=-\infty}^{\infty}
\hat{K}_m(t)e^{\mathrm{i}mx},
\label{eq:kernel_complex}
\end{align}
\end{subequations}
where $\mathrm{i}$
denotes the imaginary unit, while
\(
\hat{\rho}_m(t)
\)
and
\(
\hat{K}_m(t)
\)
are the complex Fourier coefficients of the density and kernel, respectively. Since
\(
\rho(x,t)
\) and \(
K(x,t)
\)
are real-valued, the Fourier coefficients satisfy the Hermitian symmetry
$\hat{\rho}_{-m}(t)=\overline{\hat{\rho}_{m}(t)}$
and
$\hat{K}_{-m}(t)=\overline{\hat{K}_{m}(t)}$,
for every $m\in\mathbb{Z}$.
Further, mass normalization in
\eqref{eq:mass}
implies
\begin{equation}
\hat{\rho}_0
=
\frac{1}{2\pi}.
\label{eq:rho0}
\end{equation}

Substituting the Fourier expansion
\eqref{eq:fourier_expansions}
into the velocity in
\eqref{eq:velocity},
and applying the convolution theorem \cite{Grafakos} yields 
\begin{equation}
v
=
2\pi
\sum_{m=-\infty}^{\infty}
\hat{K}_{m}(t)
\hat{\rho}_{m}(t)
e^{\mathrm{i}mx},
\label{eq:velocity_complex}
\end{equation}
from which
\begin{equation}
\rho v
=
2\pi
\sum_{m,n=-\infty}^{\infty}
\hat{\rho}_{m}(t)
\hat{\rho}_{n}(t)
\hat{K}_{n}(t)
e^{\mathrm{i}(m+n)x}.
\label{eq:rho_v}
\end{equation}
Substituting
\eqref{eq:rho_v} into the continuity equation \eqref{eq:pde} and applying orthogonality conditions yields 
\begin{equation}
\dot{\hat{\rho}}_{k}(t)
=
-2\pi \mathrm{i} k
\sum_{n=-\infty}^{\infty}
\hat{\rho}_{k-n}(t)
\hat{\rho}_{n}(t)
\hat{K}_{n}(t),
\qquad
k\in\mathbb{Z}.
\label{eq:modal_dynamics}
\end{equation}
Equation \eqref{eq:modal_dynamics} provides an exact
infinite-dimensional representation of the continuity
equation \eqref{eq:pde} in terms of the evolution of the Fourier coefficients of density and kernel, capturing the nonlinear coupling among Fourier modes.

\subsection{Design of equilibrium kernels}

We seek to design an interaction kernel $K^\mathrm{d}(x)$ that makes a desired density $\rho^\mathrm{d}(x)$ an equilibrium of \eqref{eq:modal_dynamics}. We denominate $ K^\mathrm{d}$ an \textquotedblleft{}equilibrium kernel\textquotedblright{}. Let
\(
\{\hat{\rho}^\mathrm{d}_{k}\}_{k\in\mathbb Z}
\)
denote the Fourier coefficients of the desired equilibrium density, which satisfy
$\dot{\hat{\rho}}^{\mathrm d}_{k}=0$ $(k\in\mathbb{Z})$ with corresponding equilibrium velocity 
$v^{\mathrm d}(x)=0$ $(x\in\mathbb{S}^{1})$. Let $\left\{\hat K^\mathrm{d}_k\right\}_{k\in\mathbb{Z}}$ be the Fourier coefficients of an equilibrium kernel for density $\rho^\mathrm{d}$. 
\begin{proposition}[Design of equilibrium kernels]
\label{prop:modal_equilibrium}
Let $\{\hat{K}^\mathrm{d}_k\}_{k\in \mathbb{Z}}$ be the Fourier coefficients of an equilibrium kernel for a desired density with Fourier coefficients
\(
\{\hat{\rho}^\mathrm{d}_{k}\}_{k\in\mathbb Z}
\). Then
\begin{equation}
\hat K^\mathrm{d}_{k}\hat{\rho}^\mathrm{d}_{k}=0,\qquad k\in \mathbb{Z}.
\label{eq:modal_equilibrium}
\end{equation}
\end{proposition}

\begin{proof}
At equilibrium, the convolution of $K^\mathrm{d}$ and $\rho^\mathrm{d}$, generating $v^\mathrm{d}$, must be zero. By the convolution theorem \cite{Grafakos}, the corresponding product of the Fourier coefficients of $K^\mathrm{d}$ and $\rho^\mathrm{d}$ must be zero, providing condition ~\eqref{eq:modal_equilibrium}.
\end{proof}
Equation~\eqref{eq:modal_equilibrium} establishes a fundamental
complementarity relationship between the Fourier coefficients of
the equilibrium density and those of an equilibrium
kernel. Note that infinite equilibrium kernels are possible. Equilibrium kernels with larger norms correspond to stronger interactions between agents. This modal complementarity plays a central role in our theory, providing a way to design kernels that make desired density distributions to be equilibria solutions of ~\eqref{eq:modal_dynamics}.

\subsection{Finite-dimensional, reduced-order dynamic model}

The nonlinear modal dynamics~\eqref{eq:modal_dynamics} constitutes an
infinite-dimensional system involving all Fourier modes.
However, only a finite spatial bandwidth can be meaningfully estimated
and controlled in practical conditions. First, for sufficiently smooth density fields, spatial variations are dominated by
low-wavenumber Fourier modes, with contributions of higher-order modes rapidly decaying with the wavenumber, according to well-known convergence theorems of the Fourier series of smooth functions
\cite{Trefethen2000}. Second, density fields are typically
reconstructed from discrete agent locations using kernel density
estimation, which naturally attenuates
high-wavenumber components \cite{Silverman1986,Scott2015}, acting as a spatial filter. Finally,
the finite spatial resolution of sensing and communication, along with the size of the robots, impose an
upper bound on the resolvable spatial wavenumbers, according to
classical sampling theory \cite{Oppenheim2010}. In this vein, practical applications are typically concerned only with low-wavenumber modes of the density, rather than with sharp density variations related to higher-order modes. Consequently, we approximate the nonlinear modal dynamics in~\eqref{eq:modal_dynamics} to a finite spatial bandwidth,  preserving the
dominant interaction dynamics relevant for analysis and controller
design. Accordingly, let
$\mathcal{M}_0=\{-N,\ldots,-1,0,1,\ldots,N\}$
denote the set of retained Fourier modes, where
$N\in\mathbb{N}$ is the truncation order. Let
$
\mathcal M
=
\{-N,\ldots,-1,1,\ldots,N\}
$
denote the set of retained non-zero Fourier modes. 
The density and kernel are approximated as
\begin{subequations}
\label{eq:truncated_expansions}
\begin{align}
\rho_N(x,t)
&\approx
\hat{\rho}_0
+
\sum_{k\in\mathcal{M}}
\hat{\rho}_k(t)e^{\mathrm{i}kx},
\label{eq:rho_truncated}
\\
K_N(x,t)
&\approx
\sum_{k\in\mathcal{M}_0}
\hat{K}_{k}(t) e^{\mathrm{i}kx}.
\label{eq:kernel_truncated}
\end{align}
\end{subequations}
Substituting the approximate expansions in \eqref{eq:truncated_expansions} into
\eqref{eq:modal_dynamics} yields the finite-dimensional, approximated modal dynamics

\begin{equation}
\dot{\hat{\rho}}_{k}(t)
=
-2\pi \mathrm{i} k
\sum_{n\in\mathcal{M}_0}
\hat{\rho}_{k-n}(t)
\hat{\rho}_{n}(t)
\hat{K}_{n}(t),
\qquad
k\in\mathcal{M}_0,
\label{eq:truncated_modal}
\end{equation}
where coefficients with indices outside \(\mathcal{M}_0\) are set to zero. This truncated representation provides a finite-dimensional approximation of the original mean-field dynamics.

\subsection{Linearization around equilibria}

 To design appropriate modifications of the interactions between agents in the proximity of equilibrium,
we linearize the system around an equilibrium $\rho^\mathrm{d}$ with equilibrium kernel $ K^\mathrm{d}$. Perturbations are introduced around the equilibrium in the Fourier domain according to
\begin{equation}
\hat{\rho}_k(t)
=
\hat{\rho}_k^\mathrm{d}
+
\delta\hat{\rho}_k(t),
\qquad
\hat K_k(t)
=
\hat K_k^\mathrm{d}
+
\delta\hat K_k(t),
\qquad
k\in\mathcal M_0,
\label{eq:controlled_perturbation}
\end{equation}
where \(\delta\hat{\rho}_k\) and \(\delta\hat K_k\) denote perturbations of the density and kernel coefficients, respectively. The kernel perturbation \(\delta\hat K_k\) is treated as the control input.
Substituting \eqref{eq:controlled_perturbation} into the truncated nonlinear dynamics \eqref{eq:truncated_modal} and using \eqref{eq:modal_equilibrium} yields
\begin{equation}
\delta\dot{\hat\rho}_k
=
-2\pi \mathrm{i} k
\sum_{n\in\mathcal M_0}
\hat\rho_{k-n}^\mathrm{d}
\left(
\hat K_n^\mathrm{d}\,
\delta\hat\rho_n
+
\hat\rho_n^\mathrm{d}\,
\delta\hat K_n
\right),
\qquad
k\in\mathcal M_0.
\label{eq:controlled_linearization}
\end{equation}
We introduce the vectors
\begin{subequations}
\label{eq:equilibrium_vectors}
\begin{align}
\hat{\boldsymbol{\rho}}_0^{\mathrm d}
&=
\left[
\hat{\rho}_{-N}^{\mathrm d},
\ldots,
\hat{\rho}_{0}^{\mathrm d},
\ldots,
\hat{\rho}_{N}^{\mathrm d}
\right]^{\mathrm T},
\label{eq:rho0_vector}
\\
\hat{\boldsymbol{\rho}}^{\mathrm d}
&=
\left[
\hat{\rho}_{-N}^{\mathrm d},
\ldots,
\hat{\rho}_{-1}^{\mathrm d},
\hat{\rho}_{1}^{\mathrm d},
\ldots,
\hat{\rho}_{N}^{\mathrm d}
\right]^{\mathrm T},
\label{eq:rho_vector}
\\
\hat{\mathbf K}^{\mathrm d}
&=
\left[
\hat{K}_{-N}^{\mathrm d},
\ldots,
\hat{K}_{-1}^{\mathrm d},
\hat{K}_{1}^{\mathrm d},
\ldots,
\hat{K}_{N}^{\mathrm d}
\right]^{\mathrm T},
\label{eq:K_vector}
\end{align}
\end{subequations}
along with the matrix
\begin{equation}
\mathbf D_0
=
\operatorname{diag}
\left(
-N,\ldots,-1,0,1,\ldots,N
\right),
\label{eq:D0}
\end{equation}
and the Toeplitz convolution matrix
\begin{equation}
\begingroup
\setlength{\arraycolsep}{3pt}
\renewcommand{\arraystretch}{0.92}
\mathbf T_0\!\left(\hat{\boldsymbol{\rho}}_0^{\mathrm d}\right)
=
\begin{bmatrix}
\hat{\rho}_0^{\mathrm d} &
\hat{\rho}_{-1}^{\mathrm d} &
\cdots &
\hat{\rho}_{-N}^{\mathrm d} &
0 &
\cdots &
0
\\
\hat{\rho}_{1}^{\mathrm d} &
\hat{\rho}_{0}^{\mathrm d} &
\cdots &
\hat{\rho}_{-N+1}^{\mathrm d} &
\hat{\rho}_{-N}^{\mathrm d} &
\ddots &
\vdots
\\
\vdots &
\vdots &
\ddots &
\vdots &
\vdots &
\ddots &
0
\\
\hat{\rho}_{N}^{\mathrm d} &
\hat{\rho}_{N-1}^{\mathrm d} &
\cdots &
\hat{\rho}_{0}^{\mathrm d} &
\hat{\rho}_{-1}^{\mathrm d} &
\cdots &
\hat{\rho}_{-N}^{\mathrm d}
\\
0 &
\hat{\rho}_{N}^{\mathrm d} &
\cdots &
\hat{\rho}_{1}^{\mathrm d} &
\hat{\rho}_{0}^{\mathrm d} &
\cdots &
\hat{\rho}_{-N+1}^{\mathrm d}
\\
\vdots &
\ddots &
\ddots &
\vdots &
\vdots &
\ddots &
\vdots
\\
0 &
\cdots &
0 &
\hat{\rho}_{N}^{\mathrm d} &
\hat{\rho}_{N-1}^{\mathrm d} &
\cdots &
\hat{\rho}_{0}^{\mathrm d}
\end{bmatrix}.
\endgroup
\label{eq:T0}
\end{equation}
With a slight abuse of notation, the matrix entries are indexed by the associated Fourier modes rather than by their numerical row and column indices, so that entry $(1,1)$ of the matrix is referred to as entry $(-N,-N)$, the difference being a uniform shift.
The linearized dynamics can then be written compactly as
\begin{subequations}\label{eq:linearized_system}
\begin{align}
\delta\dot{\hat{\boldsymbol{\rho}}}_0
&=
\mathbf A_0\,
\delta\hat{\boldsymbol{\rho}}_0
+
\mathbf B_0\,
\delta\hat{\mathbf K}_0,
\label{eq:full_linearized}
\\
\mathbf A_0
&=
-2\pi \mathrm{i}\,
\mathbf D_0\,
\mathbf T_0\left(\hat{\boldsymbol{\rho}}_0^{\mathrm d}\right)\,
\operatorname{diag}\left(\hat{\mathbf K}_0^{\mathrm d}\right),
\label{eq:Abar}
\\
\mathbf B_0
&=
-2\pi \mathrm{i}\,
\mathbf D_0\,
\mathbf T_0\left(\hat{\boldsymbol{\rho}}_0^{\mathrm d}\right)\,
\operatorname{diag}\left(\hat{\boldsymbol{\rho}}_0^{\mathrm d}\right).
\label{eq:Bbar}
\end{align}
\end{subequations}
Due to the zero row of $\mathbf D_0$,~\eqref{eq:linearized_system}
implies $\delta\dot{\hat\rho}_0=0$. Consequently, the perturbation of
the zero Fourier mode of the density remains constant over time. We
focus on mass-preserving perturbations satisfying
$\delta\hat\rho_0(0)=0$, which gives
$
\delta\hat\rho_0(t)=0,
\forall\,t\geq0.
$
Upon defining
\begin{equation}
\delta\hat{\boldsymbol\rho}
=
\begin{bmatrix}
\delta\hat\rho_{-N}
&
\cdots
&
\delta\hat\rho_{-1}
&
\delta\hat\rho_{1}
&
\cdots
&
\delta\hat\rho_{N}
\end{bmatrix}^{\mathrm T},
\label{eq:reduced_delta_rho}
\end{equation}
the model reduces to
\begin{equation}
\delta\dot{\hat{\boldsymbol\rho}}
=
\mathbf A\,
\delta\hat{\boldsymbol\rho}
+
\mathbf B\,
\delta\hat{\boldsymbol K}_0,
\label{eq:state_space}
\end{equation}
where
$\mathbf A\in\mathbb C^{2N\times2N}$ and
$\mathbf B\in\mathbb C^{2N\times(2N+1)}$ have entries
\begin{subequations}
\label{eq:reduced_matrix_entries}
\begin{align}
\left(\mathbf A\right)_{mn}
&=
-2\pi\mathrm{i}\,m\,
\hat{\rho}_{m-n}^{\mathrm d}
\hat K_{n}^{\mathrm d},
\qquad
m,n\in\mathcal M,
\label{eq:AR_entry}
\\
\left(\mathbf B\right)_{mn}
&=
-2\pi\mathrm{i}\,m\,
\hat{\rho}_{m-n}^{\mathrm d}
\hat{\rho}_{n}^{\mathrm d},
\qquad
m\in\mathcal M,\quad n\in\mathcal M_0,
\label{eq:BR_entry}
\end{align}
\end{subequations}
To obtain compact matrix representations of
$\mathbf A$ and $\mathbf B$, we define the reduced equilibrium vectors
\begin{subequations}\label{eq:reduced_equilibrium_vectors}
\begin{align}
\hat{\boldsymbol{\rho}}^{\mathrm d}
&=
\begin{bmatrix}
\hat{\rho}_{-N}^{\mathrm d} &
\cdots &
\hat{\rho}_{-1}^{\mathrm d} &
\hat{\rho}_{1}^{\mathrm d} &
\cdots &
\hat{\rho}_{N}^{\mathrm d}
\end{bmatrix}^{\mathrm T},
\label{eq:rhoR}
\\
\hat{\mathbf K}^{\mathrm d}
&=
\begin{bmatrix}
\hat{K}_{-N}^{\mathrm d} &
\cdots &
\hat{K}_{-1}^{\mathrm d} &
\hat{K}_{1}^{\mathrm d} &
\cdots &
\hat{K}_{N}^{\mathrm d}
\end{bmatrix}^{\mathrm T},
\label{eq:KR}
\end{align}
\end{subequations}
together with
\begin{equation}
\mathbf D
=
\operatorname{diag}
\left(
-N,\ldots,-1,1,\ldots,N
\right).
\label{eq:D}
\end{equation}
Since the equilibrium condition~\eqref{eq:modal_equilibrium}
implies
$\hat K_0^{\mathrm d}=0$
(because
$\hat\rho_0^{\mathrm d}=1/(2\pi)\neq0$, see \eqref{eq:rho0}),
the zero Fourier mode is excluded from
$\hat{\mathbf K}^{\mathrm d}$. Let
$\mathbf T_A\!\left(\hat{\boldsymbol{\rho}}_0^{\mathrm d}\right)\in\mathbb{C}^{2N\times2N}$
denote the matrix obtained from
$\mathbf T_0\!\left(\hat{\boldsymbol{\rho}}_0^{\mathrm d}\right)$
by removing the row and column associated with the zero Fourier mode.
Likewise, let
$\mathbf T_B\!\left(\hat{\boldsymbol{\rho}}_0^{\mathrm d}\right)
\in
\mathbb{C}^{2N \times (2N+1)}$
denote the matrix obtained by removing only the row associated with the zero Fourier mode.
Notice that neither
$\mathbf T_A\!\left(\hat{\boldsymbol{\rho}}_0^{\mathrm d}\right)$
nor
$\mathbf T_B\!\left(\hat{\boldsymbol{\rho}}_0^{\mathrm d}\right)$
coincides with the Toeplitz matrix generated  from
$\hat{\boldsymbol{\rho}}^{\mathrm d}$,
since
$\hat{\rho}_0^{\mathrm d}$
continues to appear in their entries, and they are not Toeplitz.
The matrices in~\eqref{eq:state_space} can therefore be written compactly as
\begin{subequations}\label{eq:factorized_matrices}
\begin{align}
\mathbf A
&=
-2\pi \mathrm{i}\,
\mathbf D\,
\mathbf T_A\left(\hat{\boldsymbol{\rho}}_0^{\mathrm d}\right)
\operatorname{diag}
\left(
\hat{\mathbf K}^{\mathrm d}
\right),
\label{eq:AR_factorized}
\\
\mathbf B
&=
-2\pi \mathrm{i}\,
\mathbf D\,
\mathbf T_B\left(\hat{\boldsymbol{\rho}}_0^{\mathrm d}\right)
\operatorname{diag}
\left(
\hat{\boldsymbol{\rho}}_0^{\mathrm d}
\right).
\label{eq:BR_factorized}
\end{align}
\end{subequations}
The factorizations in \eqref{eq:factorized_matrices} explicitly separate the
contributions of the equilibrium kernel and the equilibrium density.
Matrices
$\mathbf T_A\!\left(\hat{\boldsymbol{\rho}}_0^{\mathrm d}\right)$
and
$\mathbf T_B\!\left(\hat{\boldsymbol{\rho}}_0^{\mathrm d}\right)$
capture the convolution coupling among the retained Fourier modes, whereas the
diagonal matrices isolate the modal contributions of
$\hat{\mathbf K}^{\mathrm d}$
and
$\hat{\boldsymbol{\rho}}_0^{\mathrm d}$,
respectively.
Furthermore, the equilibrium condition~\eqref{eq:modal_equilibrium}
implies that the supports of
$\hat{\mathbf K}^{\mathrm d}$
and
$\hat{\boldsymbol{\rho}}^{\mathrm d}$
are disjoint.
Consequently, for every retained Fourier mode
$n\in\mathcal M$,
at most one of
$\hat K_n^{\mathrm d}$
and
$\hat\rho_n^{\mathrm d}$
is non-zero.
Assuming that at least one of the two is different from zero (as to improve controllability), the retained Fourier modes can be partitioned into two disjoint subsets:
those contributing to
$\mathbf A$
and those contributing to
$\mathbf B$.
This partition of the retained Fourier modes constitutes the key structural property underlying the controllability analysis developed in the next section.
\subsection{State-feedback controller design}
\label{sec:controller_design}
The system in~\eqref{eq:state_space} represents a standard linear time-invariant system, for which we design a state-feedback controller \cite{Anderson1990}.
The control input is selected as
\begin{equation}
\delta \hat{\mathbf{K}}
=
-\mathbf{K}_\mathrm{c} \delta \hat{\boldsymbol{\rho}},
\label{eq:state_feedback}
\end{equation}
where $\mathbf{K}_\mathrm{c}\in \mathbb{C}^{(2N+1)\times 2N}$ is a gain matrix that can be designed through a variety of techniques, such as pole placement ~\cite{Kailath1980} or a linear quadratic regulator (LQR) \cite{Anderson1990}, to obtain desired behaviors of the system, such as asymptotic stability. Substituting \eqref{eq:state_feedback} into
\eqref{eq:state_space} yields the closed-loop dynamics
\begin{equation}
\delta\dot{\hat{\boldsymbol{\rho}}}
=
\left(
\mathbf A
-
\mathbf B\mathbf{K}_\mathrm{c}
\right)
\delta\hat{\boldsymbol{\rho}}.
\label{eq:closed_loop}
\end{equation}
\section{Controllability Analysis}
\label{sec:controllability}

As modulation of the kernel represents a new pathway for controlling
system dynamics, an interesting question regards whether it is possible
to completely control the linearized dynamics near the desired
equilibrium, based on the truncated Fourier representation. Studying the
controllability of general systems is challenging. Thus, in this section,
we focus on demonstrating controllability for an application-relevant condition. Complementing this effort, we numerically investigate
controllability for more general systems in 
Section~\ref{numerical_results}.

Here, we focus on demonstrating controllability for a condition that
is relevant to most applications. In practical situations, we are
interested in achieving an equilibrium up to a certain spatial bandwidth.
Such an equilibrium typically contains all the modes within such a
bandwidth. We truncate the Fourier series of the equilibrium to $N-1$
and set the remaining modes to zero. We consider a single-mode
equilibrium kernel that is zero apart from the boundary modes at $\pm N$. Thus, we
consider a desired density and an equilibrium kernel in the form
\begin{subequations}
\label{eq:support_conditions}
\begin{align}
\hat{\rho}_{\pm N}^{\mathrm d}
&=
0,
\qquad
\hat{\rho}_{j}^{\mathrm d}\neq0,
\quad |j|<N,
\label{eq:rho_support}
\\
\hat{K}_{\pm N}^{\mathrm d}
&\neq0,
\qquad
\hat{K}_{j}^{\mathrm d}=0,
\quad j\neq\pm N.
\label{eq:k_support}
\end{align}
\end{subequations}
This support configuration is consistent with the modal equilibrium
condition~\eqref{eq:modal_equilibrium}: the equilibrium density is
supported on all retained modes except the boundary modes, whereas the
equilibrium kernel is supported exclusively on the boundary modes. We aim at demonstrating the controllability of the truncated system
in~\eqref{eq:state_space}, under the specific conditions
of~\eqref{eq:support_conditions}. To this end, we seek to determine
under which conditions the controllability matrix
\begin{equation}
\mathcal C
=
\begin{bmatrix}
\mathbf B &
\mathbf A\mathbf B &
\mathbf A^2\mathbf B &
\cdots &
\mathbf A^{2N-1}\mathbf B
\end{bmatrix}
\label{eq:controllability_matrix}
\end{equation}
has full row rank~\cite{Chen1999}. First, we focus on the conditions on the desired density's Fourier
coefficients that ensure that the remaining nonzero columns of
$\mathbf B$ are linearly independent, and hence that
$\operatorname{rank}(\mathbf B)=2N-1$.
To this end, consider the full Toeplitz matrix
$\mathbf T_0\left(\hat{\boldsymbol{\rho}}_0^{\mathrm d}\right)$. For any non-zero
$\mathbf c\in\mathbb C^{2N+1}$, the associated quadratic
form admits the standard integral representation
\cite{GrenanderSzego1958}
\begin{equation}
\mathbf c^*
\mathbf T_0\left(\hat{\boldsymbol{\rho}}_0^{\mathrm d}\right)
\mathbf c
=
\frac{1}{2\pi}
\int_{-\pi}^{\pi}
\rho^{\mathrm d}(x)
\left|
\sum_{k=-N}^{N}
c_k e^{\mathrm{i}kx}
\right|^2
\,\mathrm{d}x.
\label{eq:T0_quadratic_form}
\end{equation}
Since the prescribed desired density is nonnegative and normalized,
it is positive on a set of non-zero measure, such that
$\mathbf T_0\left(\hat{\boldsymbol{\rho}}_0^{\mathrm d}\right)$ is
positive definite and thus nonsingular, with rank $2N+1$. Consequently,
\begin{equation}
\operatorname{rank}\left[
\mathbf T_0\left(\hat{\boldsymbol{\rho}}_0^{\mathrm d}\right)
\right]
=
2N+1.
\label{eq:T0_full_rank}
\end{equation}
Since
$\mathbf T_0\left(\hat{\boldsymbol{\rho}}_0^{\mathrm d}\right)$
is nonsingular, removing the two columns indexed by $\pm N$ leaves
$2N-1$ linearly independent columns. Additionally removing the
zero-mode row can reduce this rank by at most one. Therefore, the
resulting matrix has rank of either $2N-1$ or $2N-2$. It remains to
determine when removal of the zero-mode row preserves the rank
$2N-1$. To characterize the maximal minors of
$\widetilde{\mathbf T}_B
\left(\hat{\boldsymbol{\rho}}_0^{\mathrm d}\right)$,
we note that these are also minors of
$\mathbf T_0
\left(\hat{\boldsymbol{\rho}}_0^{\mathrm d}\right)$,
obtained by deleting the zero-mode row, one additional row, and the
columns indexed by $\pm N$. We employ a finite banded-Toeplitz
minor representation in terms of skew Schur
polynomials~\cite{Maximenko2017}. Consider the Laurent polynomial
generated by the desired density's coefficients,
\begin{equation}
a(z)
=
\sum_{k=-(N-1)}^{N-1}
\hat{\rho}_{k}^{\mathrm d}z^k,
\qquad
z^{N-1}a(z)
=
\hat{\rho}_{N-1}^{\mathrm d}
\prod_{\ell=1}^{2N-2}(z-z_\ell),
\label{eq:density_symbol}
\end{equation}
where $z_1,\ldots,z_{2N-2}$ denote the roots associated with the
desired density's Fourier coefficients. For each $r\in\mathcal M$, we define the partitions
\begin{equation}
\boldsymbol{\beta}^{(r)}
=
\begin{cases}
(N-r-1,N), & r<0,\\[1mm]
(N-1,N-r), & r>0,
\end{cases}
\qquad
\boldsymbol{\alpha}
=
((2N-1)^N).
\label{eq:schur_partitions}
\end{equation}
Here, a partition is a finite non-increasing sequence of nonnegative
integers. The partitions $\boldsymbol{\alpha}$ and
$\boldsymbol{\beta}^{(r)}$ determine the corresponding skew Schur
polynomial
$s_{\boldsymbol{\alpha}/\boldsymbol{\beta}^{(r)}}$, which is a
symmetric polynomial in $z_1,\ldots,z_{2N-2}$, providing a compact representation of the
maximal minors of the banded Toeplitz matrix.
\begin{lemma}
\label{lem:Rank_BR}
Let $N\ge2$. Under~\eqref{eq:support_conditions}, the matrix
$\mathbf B$ in~\eqref{eq:reduced_matrix_entries} satisfies
\begin{equation}
\operatorname{rank}(\mathbf B)=2N-1
\label{eq:rank_BR}
\end{equation}
if and only if there exists at least one $r\in\mathcal M$ such that
\begin{equation}
s_{\boldsymbol{\alpha}/\boldsymbol{\beta}^{(r)}}
\left(z_1,\ldots,z_{2N-2}\right)
\neq0,
\label{eq:schur_rank_condition}
\end{equation}
where
$s_{\boldsymbol{\alpha}/\boldsymbol{\beta}^{(r)}}$
denotes the skew Schur polynomial associated with the partitions in
\eqref{eq:schur_partitions}.
\end{lemma}

\begin{proof}
Let
$\widetilde{\mathbf B}\in
\mathbb C^{2N\times(2N-1)}$
denote the matrix obtained from $\mathbf B$ by removing the columns
indexed by $n=\pm N$. Under~\eqref{eq:rho_support}, these columns
vanish by~\eqref{eq:BR_entry}, since
$\hat{\rho}_{\pm N}^{\mathrm d}=0$. Hence,
$\operatorname{rank}(\mathbf B)
=
\operatorname{rank}(\widetilde{\mathbf B})$.
From the factorization in~\eqref{eq:BR_factorized},
\begin{equation}
\widetilde{\mathbf B}
=
-2\pi\mathrm{i}\,
\mathbf D\,
\widetilde{\mathbf T}_B
\left(
\hat{\boldsymbol{\rho}}_0^{\mathrm d}
\right)
\widetilde{\mathbf D}_{\rho},
\label{eq:B_tilde_factorization}
\end{equation}
where
$\widetilde{\mathbf T}_B
\left(\hat{\boldsymbol{\rho}}_0^{\mathrm d}\right)
\in\mathbb C^{2N\times(2N-1)}$
is obtained from
$\mathbf T_B\left(\hat{\boldsymbol{\rho}}_0^{\mathrm d}\right)$
by removing the columns indexed by $\pm N$, and
\begin{equation}
\widetilde{\mathbf D}_{\rho}
=
\operatorname{diag}
\left(
\hat{\rho}_{-N+1}^{\mathrm d},
\ldots,
\hat{\rho}_{-1}^{\mathrm d},
\hat{\rho}_{0}^{\mathrm d},
\hat{\rho}_{1}^{\mathrm d},
\ldots,
\hat{\rho}_{N-1}^{\mathrm d}
\right).
\label{eq:Drho_tilde}
\end{equation}
The matrix $\mathbf D$ is nonsingular because its diagonal entries
are the non-zero Fourier modes in $\mathcal M$, whereas
$\widetilde{\mathbf D}_{\rho}$ is nonsingular by
\eqref{eq:rho_support}. Hence,
\begin{equation}
\operatorname{rank}(\mathbf B)
=
\operatorname{rank}(\widetilde{\mathbf B})
=
\operatorname{rank}\left[
\widetilde{\mathbf T}_B
\left(\hat{\boldsymbol{\rho}}_0^{\mathrm d}\right)
\right].
\label{eq:B_TB_rank_equivalence}
\end{equation}
Since
$\widetilde{\mathbf T}_B
\left(\hat{\boldsymbol{\rho}}_0^{\mathrm d}\right)$
has $2N-1$ columns, it has rank $2N-1$ if and only if at least one
of its maximal $(2N-1)\times(2N-1)$ minors is nonsingular. For each
$r\in\mathcal M$, let $\mathbf M_r$ denote the maximal minor obtained
by removing the row indexed by $r$. Hence,
\begin{equation}
\operatorname{rank}\left[
\widetilde{\mathbf T}_B
\left(
\hat{\boldsymbol{\rho}}_0^{\mathrm d}
\right)
\right]
=
2N-1
\iff
\exists\,r\in\mathcal M:
\det(\mathbf M_r)\neq0.
\label{eq:maximal_minor_condition}
\end{equation}
Each $\mathbf M_r$ is also a minor of the full Toeplitz matrix
$\mathbf T_0\left(\hat{\boldsymbol{\rho}}_0^{\mathrm d}\right)$,
obtained by deleting the rows associated with the modes $0$ and $r$
and the columns associated with the modes $\pm N$. The Fourier-mode
index $q\in\{-N,\ldots,N\}$ corresponds to the one-based Toeplitz
index $q+N+1$. Therefore, the deleted column indices are $(1,2N+1)$,
whereas the ordered deleted row indices are
\begin{equation}
\boldsymbol{\xi}^{(r)}
=
\begin{cases}
(r+N+1,N+1), & r<0,\\[1mm]
(N+1,r+N+1), & r>0.
\end{cases}
\label{eq:deleted_row_indices}
\end{equation}
The matrix
$\mathbf T_0\left(\hat{\boldsymbol{\rho}}_0^{\mathrm d}\right)$
is a finite banded Toeplitz matrix generated by the Laurent polynomial
$a(z)$ in~\eqref{eq:density_symbol}. Applying the finite banded-Toeplitz minor representation in
Theorem~2.1 of~\cite{Maximenko2017}, we have
\begin{equation}
\bigl|\det(\mathbf M_r)\bigr|
=
\bigl|\hat{\rho}_{N-1}^{\mathrm d}\bigr|^{2N-1}
\bigl|
s_{\boldsymbol{\alpha}/\boldsymbol{\beta}^{(r)}}
(z_1,\ldots,z_{2N-2})
\bigr|.
\label{eq:minor_schur_representation}
\end{equation}
The deleted columns $(1,2N+1)$ yield
$\boldsymbol{\alpha}=((2N-1)^N)$. For the deleted rows, the
corresponding partition in Theorem~2.1 of~\cite{Maximenko2017} is
$\left(2N-\xi_1^{(r)},\,2N+1-\xi_2^{(r)}\right)$.
Substitution of~\eqref{eq:deleted_row_indices} gives
$\boldsymbol{\beta}^{(r)}$ in~\eqref{eq:schur_partitions}.
Since
$\bigl|\hat{\rho}_{N-1}^{\mathrm d}\bigr|>0$
by~\eqref{eq:rho_support},
\eqref{eq:minor_schur_representation} implies
\begin{equation}
\det(\mathbf M_r)\neq0
\quad\Longleftrightarrow\quad
s_{\boldsymbol{\alpha}/\boldsymbol{\beta}^{(r)}}
(z_1,\ldots,z_{2N-2})
\neq0.
\label{eq:minor_schur_equivalence}
\end{equation}
Combining \eqref{eq:B_TB_rank_equivalence},
\eqref{eq:maximal_minor_condition}, and
\eqref{eq:minor_schur_equivalence} proves
\eqref{eq:rank_BR}.
\end{proof}
Under~\eqref{eq:schur_rank_condition}, $\mathbf B$ has rank
deficiency one. The condition is computationally inexpensive to check
for a given truncation order $N$, as it requires computing the
$2N-2$ roots of the polynomial in~\eqref{eq:density_symbol} and
evaluating the corresponding skew Schur polynomials for
$r\in\mathcal M$. Since only one of these polynomials needs to be
nonzero, the evaluation can be terminated as soon as a nonzero value
is found. For the controllability matrix $\mathcal C$ to be full-rank ($2N$),
we only require that the combinations $\mathbf A^i\mathbf B$ generate
a column vector outside $\operatorname{Col}(\mathbf B)$. We show that,
under~\eqref{eq:support_conditions}, $\mathbf A\mathbf B$ is enough to
generate such an independent column. To this end, we exploit the boundary-mode structure imposed by
\eqref{eq:support_conditions}. Under these conditions, $\mathbf A$
acts only through its two boundary columns, so that
$\mathbf A\mathbf B$ is determined by the corresponding boundary rows
of $\mathbf B$. We show that these boundary contributions are not
annihilated by the one-dimensional left nullspace of $\mathbf B$,
allowing $\mathbf A\mathbf B$ to recover the single column missing
from $\mathbf B$. 
\begin{lemma}
\label{lem:boundary_row}
Let
\begin{equation}
\mathbf V
=
\begin{bmatrix}
\mathbf B(-N,:)\\
\mathbf B(N,:)
\end{bmatrix}
\in\mathbb C^{2\times(2N+1)}.
\label{eq:V_def}
\end{equation}
Under~\eqref{eq:rho_support}, for $N\ge2$,
\begin{equation}
\operatorname{rank}\left(\mathbf V\right)=2.
\label{eq:rank_V}
\end{equation}
\end{lemma}
\begin{proof}
For $n\in\mathcal M_0$, the entries of the two boundary rows 
from~\eqref{eq:BR_entry} are
\begin{subequations}
\label{eq:boundary_rows}
\begin{align}
(\mathbf V)_{1,n}
&=
2\pi\mathrm{i}N\,
\hat{\rho}_n^{\mathrm d}
\hat{\rho}_{-N-n}^{\mathrm d},
\label{eq:minus_boundary_row}
\\
(\mathbf V)_{2,n}
&=
-2\pi\mathrm{i}N\,
\hat{\rho}_n^{\mathrm d}
\hat{\rho}_{N-n}^{\mathrm d}.
\label{eq:plus_boundary_row}
\end{align}
\end{subequations}
By~\eqref{eq:rho_support}, the first row of $\mathbf V$ can have non-zero entries only for
$n\in\{-N+1,\ldots,-1\}$, while the second row can have non-zero
entries only for $n\in\{1,\ldots,N-1\}$.
Thus, the two rows of $\mathbf V$ have disjoint supports and are non-zero due to~\eqref{eq:rho_support}, making them linearly independent, so that
$\operatorname{rank}(\mathbf V)=2$.
\end{proof}
Under~\eqref{eq:support_conditions}, the only non-zero columns of
$\mathbf A$ are those indexed by $\pm N$, denoted by
$\mathbf a_{\pm N}$. Let
\begin{equation}
\mathbf A_N
=
\begin{bmatrix}
\mathbf a_{-N} & \mathbf a_N
\end{bmatrix}
\in\mathbb C^{2N\times2},
\label{eq:AN_definition}
\end{equation}
which selects the non-zero columns of $\mathbf A$.
\begin{lemma}
\label{lem:missing_direction}
Suppose that~\eqref{eq:support_conditions} and~\eqref{eq:schur_rank_condition}
hold. Let $\mathbf q\in\mathbb C^{2N\times1}$ be a non-zero vector
spanning the one-dimensional left nullspace of $\mathbf B$, that is,
\begin{equation}
\mathbf q^*\mathbf B=\mathbf 0,
\label{eq:q_left_null}
\end{equation}
where $(\cdot)^*$ denotes the conjugate transpose. Then
\begin{equation}
\mathbf q^*\mathbf A_N\neq\mathbf 0.
\label{eq:q_AN_non-zero}
\end{equation}
\end{lemma}

\begin{proof}
Since
$\mathbf T_0\left(\hat{\boldsymbol{\rho}}_0^{\mathrm d}\right)$
is nonsingular by~\eqref{eq:T0_full_rank}, deleting its zero-mode
row yields
$\mathbf T_B\left(\hat{\boldsymbol{\rho}}_0^{\mathrm d}\right)
\in\mathbb C^{2N\times(2N+1)}$
with full row rank. Hence,
\begin{equation}
\operatorname{rank}\left[
\mathbf T_B\left(\hat{\boldsymbol{\rho}}_0^{\mathrm d}\right)
\right]
=
2N.
\label{eq:TB_full_row_rank}
\end{equation}
Let $\mathbf t_{-N}$ and $\mathbf t_N$ denote the columns of
$\mathbf T_B\left(\hat{\boldsymbol{\rho}}_0^{\mathrm d}\right)$
indexed by $-N$ and $N$, respectively. Then
\begin{equation}
\mathbf T_B\left(\hat{\boldsymbol{\rho}}_0^{\mathrm d}\right)
=
\begin{bmatrix}
\mathbf t_{-N} &
\widetilde{\mathbf T}_B
\left(\hat{\boldsymbol{\rho}}_0^{\mathrm d}\right) &
\mathbf t_N
\end{bmatrix}.
\label{eq:TB_column_partition}
\end{equation}
From~\eqref{eq:B_tilde_factorization} and
\eqref{eq:q_left_null},
\begin{equation}
\mathbf q^*\mathbf D\,
\widetilde{\mathbf T}_B
\left(\hat{\boldsymbol{\rho}}_0^{\mathrm d}\right)
=
\mathbf 0,
\label{eq:qD_TB_tilde_zero}
\end{equation}
since $\widetilde{\mathbf D}_{\rho}$ is nonsingular by
\eqref{eq:rho_support}. Then, at least one between
$\mathbf q^*\mathbf D\mathbf t_{-N}$ and
$\mathbf q^*\mathbf D\mathbf t_N$
must be non-zero. In fact, if both were zero vectors,
\eqref{eq:TB_full_row_rank} would imply that
$\mathbf q^*\mathbf D=\mathbf 0$ and, since $\mathbf D$ is
nonsingular, we would find $\mathbf q=\mathbf 0$, which contradicts
the definition of $\mathbf q$.
Under~\eqref{eq:k_support}, the two non-zero columns of $\mathbf A$
satisfy
\begin{equation}
\mathbf A_N
=
-2\pi\mathrm{i}\,
\mathbf D
\begin{bmatrix}
\mathbf t_{-N} & \mathbf t_N
\end{bmatrix}
\operatorname{diag}\left(
\hat K_{-N}^{\mathrm d},
\hat K_N^{\mathrm d}
\right).
\label{eq:AN_TB_factorization}
\end{equation}
Premultiplying~\eqref{eq:AN_TB_factorization} by $\mathbf q^*$
gives
\begin{equation}
\mathbf q^*\mathbf A_N
=
-2\pi\mathrm{i}\,
\begin{bmatrix}
\mathbf q^*\mathbf D\mathbf t_{-N} &
\mathbf q^*\mathbf D\mathbf t_N
\end{bmatrix}
\operatorname{diag}\left(
\hat K_{-N}^{\mathrm d},
\hat K_N^{\mathrm d}
\right).
\label{eq:q_AN_factorization}
\end{equation}
Since $\hat K_{\pm N}^{\mathrm d}\neq0$ by
\eqref{eq:k_support}, the condition that at least one between
$\mathbf q^* \mathbf D \mathbf t_{-N}$ and
$\mathbf q^* \mathbf D \mathbf t_N$ is non-zero yields
$\mathbf q^*\mathbf A_N\neq\mathbf 0$.
\end{proof}
\begin{lemma}
\label{lem:AB_recovery}
Assume~\eqref{eq:support_conditions} and~\eqref{eq:schur_rank_condition} hold. Then,
\begin{equation}
\mathbf q^*\mathbf A\mathbf B\neq\mathbf 0.
\label{eq:qAB_non-zero}
\end{equation}
\end{lemma}

\begin{proof}
Under~\eqref{eq:k_support}, the only non-zero columns of $\mathbf A$
are those indexed by $\pm N$. Therefore, using
\eqref{eq:AN_definition} and \eqref{eq:V_def},
\begin{equation}
\mathbf A\mathbf B
=
\mathbf a_{-N}\mathbf B(-N,:)
+
\mathbf a_N\mathbf B(N,:)
=
\mathbf A_N\mathbf V.
\label{eq:AB_factorization}
\end{equation}
Premultiplying~\eqref{eq:AB_factorization} by $\mathbf q^*$ gives
\begin{equation}
\mathbf q^*\mathbf A\mathbf B
=
\left(\mathbf q^*\mathbf A_N\right)\mathbf V.
\label{eq:qAB_factorization}
\end{equation}
By Lemma~\ref{lem:missing_direction},
$\mathbf q^*\mathbf A_N\neq\mathbf 0$. As $\operatorname{rank}(\mathbf V)=2$ by Lemma~\ref{lem:boundary_row},
$\mathbf V$ has a trivial left nullspace, such that
$
\left(\mathbf q^*\mathbf A_N\right)\mathbf V
\neq\mathbf 0$, which is equivalent to~\eqref{eq:qAB_non-zero}.
\end{proof}

\begin{theorem}
\label{thm:controllability}
Suppose that the density-kernel support
conditions~\eqref{eq:support_conditions} hold and that the skew Schur
nondegeneracy condition~\eqref{eq:schur_rank_condition} is satisfied.
Then, the linearized system~\eqref{eq:state_space} is controllable via
kernel modifications, whereby
\begin{equation}
\operatorname{rank}(\mathcal C)
=
2N.
\label{eq:full_controllability_rank}
\end{equation}
\end{theorem}
\begin{proof}
By Lemma~\ref{lem:Rank_BR},
$\operatorname{rank}(\mathbf B)=2N-1$, and hence its left nullspace is
one-dimensional. Let $\mathbf q$ be the non-zero vector introduced in
Lemma~\ref{lem:missing_direction}, such that
$\mathbf q^*\mathbf B=\mathbf 0$. By
Lemma~\ref{lem:AB_recovery},
$\mathbf q^*\mathbf A\mathbf B\neq\mathbf 0$, which implies that
\begin{equation}
\mathbf q^*\mathbf v_j\neq0,
\label{eq:qvj_nonzero}
\end{equation}
for at least one column $\mathbf{v}_j$ of $\mathbf{A}\mathbf{B}$, $j\in\mathcal M_0$. $\mathbf q^*\mathbf B=\mathbf 0$ requires
$\mathbf q^*\mathbf v=0$ for every
$\mathbf v\in\operatorname{Col}(\mathbf B)$, such that
\eqref{eq:qvj_nonzero} implies 
$\mathbf v_j\notin\operatorname{Col}(\mathbf B)$. Therefore, $\mathbf A\mathbf B$ contains at least one column that is
linearly independent of the columns of $\mathbf B$. Since
$\operatorname{rank}(\mathbf B)=2N-1$ and the state dimension is
$2N$,
\begin{equation}
\operatorname{rank}
\begin{bmatrix}
\mathbf B & \mathbf A\mathbf B
\end{bmatrix}
=
2N,
\label{eq:rank_B_AB_blocks}
\end{equation}
from which~\eqref{eq:full_controllability_rank} follows.
\end{proof}
\section{Results}
\label{numerical_results}

This section validates the proposed interaction modulation framework through a series of numerical inquiries. The simulations evaluate the ability of the proposed state-feedback controller to stabilize the system dynamics, assess its robustness to large perturbations with respect to the equilibrium and neglected higher-order Fourier modes, verify consistency between the reduced mean-field formulation and the corresponding microscopic multi-agent dynamics, and provide numerical support for and beyond the controllability condition established in Theorem~\ref{thm:controllability}. 

\subsection{Equilibrium stabilization}

To demonstrate the feasibility of the proposed interaction modulation framework, we consider two representative equilibrium densities on $\mathbb{S}^{1}$. The first equilibrium is a von Mises distribution represented through
a truncated Fourier expansion, whereas the second is prescribed
directly through a finite Fourier expansion. Together, these examples demonstrate that the proposed controller is applicable to analytically defined equilibrium profiles (which have typically non-sparse Fourier complex coefficients), as well as general density distributions represented in the Fourier basis, possibly reconstructed from desired values of the density at selected points on a spatial mesh.

We first consider a von Mises desired density~\cite{Maffettone}
\begin{equation}
\rho^{\mathrm d}(x)
=
\frac{C e^{\kappa\cos(x-\mu)}}{2\pi I_0(\kappa)},
\label{eq:von_mises}
\end{equation}
where $C$ is the normalization constant, $\mu$ denotes the mean,
$\kappa>0$ is the concentration parameter, and $I_0(\kappa)$ is
the modified Bessel function of the first kind of order zero.
Throughout this example, we set $C=1$, $\mu=1$, and $\kappa=2$. The desired
density is truncated to $|k|=15$ Fourier modes, while we consider $N=16$ in~\eqref{eq:state_space}, so that we set
$\hat{\rho}^{\mathrm d}_{\pm16}=0$. All coefficients $\{-15,\dots,15\}$ of the Fourier series of~\eqref{eq:von_mises} are non-zero, similar to our application-relevant controllability example in~\eqref{eq:support_conditions}. The resulting truncated
representation closely approximates the original von Mises density,
with a maximum pointwise reconstruction error of $\sim10^{-14}$.
To investigate the influence of the equilibrium kernel on the transient
closed-loop behavior, we consider two representative kernels,
\begin{subequations}
\label{eq:kernel}
\begin{align}
K^{\mathrm d}(x)
&=
0.01\cos(16x)+0.02\sin(16x),
\label{eq:kernel_mode16}
\\
K^{\mathrm d}(x)
&=
0.01\cos(x)+0.02\sin(x).
\label{eq:kernel_mode1}
\end{align}
\end{subequations}
The kernel in~\eqref{eq:kernel_mode16} satisfies the modal equilibrium condition
in~\eqref{eq:modal_equilibrium} and represents the exact application-relevant case studied for controllability in Section~\ref{sec:controllability}. Under the chosen parameters, the truncated system~\eqref{eq:state_space} resulting from~\eqref{eq:von_mises} and~\eqref{eq:kernel_mode16} satisfies the hypotheses of Theorem~\ref{thm:controllability}, making the system controllable.
The kernel in~\eqref{eq:kernel_mode1} does not satisfy the modal equilibrium condition in~\eqref{eq:modal_equilibrium}. We include this case to demonstrate that the proposed controller can still
effectively regulate the density when the kernel is not an equilibrium one.
In this case, the linearized dynamics in~\eqref{eq:state_space} contains a
nonzero constant forcing term that depends on the desired equilibrium density and kernel, and requires a nonzero kernel
modification input $\delta \hat{\mathbf{K}}$ at equilibrium.

The feedback controller in~\eqref{eq:state_feedback} is designed from the reduced
linearized model using a continuous-time
LQR, with
$
\mathbf{Q}=10^{3}\,\mathbf{I},
\mathbf{R}=\mathbf{I}.
$
For both kernels in~\eqref{eq:kernel}, the open-loop linearized dynamics is not
asymptotically stable. With kernel~\eqref{eq:kernel_mode16},
all open-loop poles are located at the origin, whereas kernel
~\eqref{eq:kernel_mode1} introduces a complex-conjugate pair with
real part $-0.01$, with the remaining poles at the origin. Under the LQR feedback, all closed-loop poles are shifted into the open negative real part half-plane, confirming asymptotic stability of the linearized closed-loop system.

To initialize the simulation, we discretize the desired density profile using 60 agents. The initial positions of the agents are perturbed using
zero-mean Gaussian independent noise with standard deviation $1$,
wrapped onto the periodic domain. The corresponding perturbed initial density $\rho(x,0)$ is constructed
from the perturbed discrete agents' positions via a periodic Gaussian kernel
density estimator~\cite{Silverman1986} with bandwidth $\sigma=0.5$, and used to find the kernel input $\delta \hat{\mathbf{K}}$ from~\eqref{eq:state_feedback}. The agent positions are then
evolved according to the microscopic dynamics in \eqref{eq:micro},
with the modified interaction kernel~\eqref{eq:controlled_perturbation}. At each time step, the
continuous density $\rho(x,t)$ is constructed from the updated discrete
configuration using the same density estimator. The kernel
is not modified in the uncontrolled case. Given the controlled and uncontrolled density solutions
$\rho_{\mathrm{c}}(x_j,t)$ and $\rho_{\mathrm{uc}}(x_j,t)$ over the uniform spatial grid
$\{x_j\}_{j=1}^{N_x}$ with spacing $\Delta x$, we define errors consistent with~\eqref{eq:error_norm} as
\begin{equation}
e_{i}(t)
=
\left[
\Delta x
\sum_{j=1}^{N_x}
\left|
\rho_{i}(x_j,t)
-
\rho^{\mathrm d}(x_j)
\right|^2
\right]^{1/2},
\label{eq:closed_loop_norm}
\end{equation}
with $i=\mathrm{c}, \mathrm{uc}$.
The density is evaluated on a uniform spatial grid of $N_x=33$ points.

The resulting density distributions and tracking errors are shown in
Fig.~\ref{fig:vonmises_results}. In the uncontrolled case, the density does not converge to the
desired profile, as shown in Fig.~\ref{fig:vonmises_results}(a), and the corresponding $L^2$
density error remains approximately constant. In contrast, the
interaction kernel feedback drives the controlled density toward the
desired von Mises profile, as shown in Fig.~\ref{fig:vonmises_results}(b), while the
corresponding $L^2$ density errors decrease rapidly for both kernels, as shown in
Fig.~\ref{fig:vonmises_results}(c). Since the density is reconstructed from a
finite number of agents using kernel density estimation (KDE), the achievable accuracy is affected by
the finite-agent representation and density reconstruction. These results
demonstrate that modulation of the interaction kernel regulates the
microscopic agent dynamics toward the prescribed desired density.
\begin{figure*}[!t]
\centering

\subfloat[Uncontrolled]{
\includegraphics[width=0.27\textwidth]{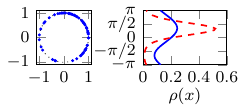}
}
\hfill
\subfloat[Controlled]{
\includegraphics[width=0.27\textwidth]{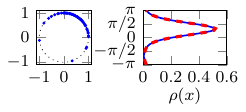}
}
\hfill
\subfloat[Error evolution]{
\includegraphics[width=0.27\textwidth]{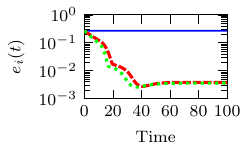}
}

\caption{
Stabilization of the von Mises equilibrium using the discrete-agent model.
(a) Uncontrolled and (b) controlled final agent configurations for the
equilibrium kernel~\eqref{eq:kernel_mode16}, together with the corresponding
density profiles reconstructed using KDE (blue). The red dashed curves denote the desired
density.
(c) Evolution of the $L^2$ density error. The solid blue curve denotes the
uncontrolled case, while the dotted green and dashed red curves indicate the
controlled cases corresponding to the equilibrium
kernels~\eqref{eq:kernel_mode16} and~\eqref{eq:kernel_mode1}, respectively.
}
\label{fig:vonmises_results}
\end{figure*}

While the previous example focuses on the conditions of the controllability theorem in Section~\ref{sec:controllability} that are relevant for most practical applications, our kernel modulation control approach could be applied even for conditions with sparse coefficients that do not satisfy~\eqref{eq:support_conditions} and for relatively large perturbations from the equilibrium. For example, we consider the desired equilibrium 
\begin{equation}
\begin{aligned}
\rho^\mathrm {d}(x)
=
\frac{1}{2\pi}
\Big(
&
1
+0.20\cos x
+0.15\sin x
\\
&
+0.10\sin 2x
+0.12\cos 4x
+0.08\sin 5x
\Big),
\end{aligned}
\label{eq:fourier_density}
\end{equation}
with the equilibrium kernel
\begin{equation}
\begin{aligned}
K^\mathrm {d}(x)
=
&
0.10\cos(3x)
+0.06\sin(3x)
\\
&
+0.08\cos(6x)
+0.04\sin(6x).
\end{aligned}
\label{eq:fourier_kernel}
\end{equation}
For this example, the Fourier dynamics is truncated at $N=6$.
We initialize the simulation to the spatially uniform density
$\rho(x,0)=1/(2\pi)$,
which differs substantially from the prescribed equilibrium.
The controlled and uncontrolled cases are simulated using a finite-volume
discretization of the nonlinear continuum dynamics in
\eqref{eq:pde}--\eqref{eq:velocity}, with a local Lax--Friedrichs
(Rusanov) numerical flux and periodic boundary conditions
\cite{LeVeque2002}. The maximum time step is
$\Delta t_{\max}=0.01$, providing a maximum Courant–Friedrichs–Lewy (CFL) number of 0.45. Figure~\ref{fig:controlled_stabilization}(a) compares the desired, initial, uncontrolled, and controlled density profiles. Starting from the uniform distribution, the uncontrolled dynamics remains far from the prescribed equilibrium throughout the simulation. In contrast, the proposed controller drives the density toward the desired profile.
Figure~\ref{fig:controlled_stabilization}(b) reports the corresponding
error from~\eqref{eq:closed_loop_norm},
evaluated on a uniform spatial grid of $N_x=800$ points.
While the error for the uncontrolled system remains nearly constant,
the error for the controlled system initially shows a slow reduction, followed by a faster decrease by more than one order
of magnitude, as the actual density becomes closer to the desired one. Although the controller is synthesized from a local linearization about the desired equilibrium, this example shows that it may  successfully regulate a large perturbation from the equilibrium, suggesting that the practical basin of attraction of equilibria may be relatively large, as shown in Section~\ref{sec:initial_perturbations}. Together, Figs.~\ref{fig:vonmises_results} and~\ref{fig:controlled_stabilization} demonstrate that interaction modulation can effectively regulate the collective dynamics around a desired equilibrium.

\begin{figure}[!t]
\centering

\subfloat[Density evolution]{
\includegraphics[width=0.40\columnwidth]{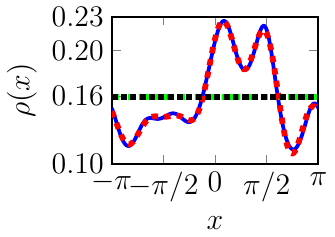}
}
\hfill
\subfloat[Error evolution]{
\includegraphics[width=0.47\columnwidth]{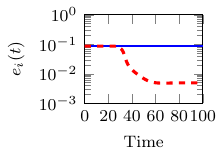}
}

\vspace{-1mm}

\caption{
Stabilization of a sparse desired density.
(a) Desired (red dashed), controlled (blue solid), initial (green dash-dotted), and uncontrolled (black dotted) density profiles.
(b) Evolution of the $L^2$ density error.
The blue solid and red dashed curves correspond to the uncontrolled
and controlled systems, respectively.
}
\label{fig:controlled_stabilization}
\end{figure}
\subsection{Robustness under large initial perturbations and modal truncation}
\label{sec:initial_perturbations}

Next, we evaluate the robustness of the
proposed controller with respect to two deviations from nominal
conditions used for controller synthesis: large initial perturbations
from the desired equilibrium and presence of additional higher-order Fourier modes, neglected in the controller. For both studies, we use the desired density and
equilibrium kernel defined in \eqref{eq:fourier_density} and
\eqref{eq:fourier_kernel}, respectively, and the feedback
controller synthesized from the reduced-order model with truncation
order $N=6$. The controlled nonlinear continuum dynamics in
\eqref{eq:pde}--\eqref{eq:velocity} is simulated using a finite-volume
method on a uniform spatial grid of $N_x=800$ points. The closed-loop
performance is evaluated using the discretized $L^2$ density error
measure in \eqref{eq:closed_loop_norm}.

To examine the effect of the initial perturbation magnitude, we perturb
the Fourier coefficients of the initial density around the desired one as
\begin{equation}
\hat{\rho}_k(0)=\hat{\rho}_k^{\mathrm d}
+\varepsilon\,\Delta\hat{\rho}_k,\qquad k\neq0, 
\label{eq:perturbation_rhok}
\end{equation}
where $\Delta\hat{\rho}_k$ defines a randomly selected perturbation direction (which remains fixed across different error scalings but varies across modes) and
$\varepsilon$ scales its magnitude. Random selection of the perturbation direction is based on the extraction of real and imaginary components from the uniform distribution over $[0,1]$. We consider
$\varepsilon\in\{10^{-1},5\times10^{-1},10^{0}\}$, corresponding to
progressively larger deviations from the desired equilibrium. The corresponding error trajectories are shown
in Fig.~\ref{fig:robustness_results}(a). Despite the different
perturbation amplitudes, all three trajectories exhibit similar
convergence behavior and reach comparable error levels over the
simulated horizon. These results indicate that the proposed feedback
controller may remain effective for a broad range of initial
perturbations.

To evaluate the influence of Fourier mode truncation, the same controller is
applied without redesign, while higher-order perturbations that are not
included in the nominal reduced-order model are introduced into the
initial density. Additional perturbations involving
the mode sets $\{6,7,8\}$ and $\{7,8,9\}$ are considered; the former partially overlaps with the spectrum of the equilibrium kernel in~\eqref{eq:fourier_density}, while the latter is beyond the order of the system. We set these perturbations as in~\eqref{eq:perturbation_rhok} using $\varepsilon=10^{-2}$. The corresponding results are presented in
Fig.~\ref{fig:robustness_results}(b). The nominal case achieves the
smallest residual error, while the additional modal perturbations lead
to higher residual error levels. Nevertheless, the error significantly decreases from
its initial value and remains bounded in both perturbed cases. We do not register strong differences between the two perturbed cases. These
results indicate that the proposed feedback strategy may provide a satisfactory
closed-loop performance in the presence of the considered higher-order
Fourier mode perturbations. We stress that these perturbations are likely not to appear in a practical application of our framework, as they will be typically filtered out by the spatial smoothing effect of KDE.

\begin{figure}[!t]
\centering
\subfloat[Large initial perturbations]{
\includegraphics[width=0.47\columnwidth]{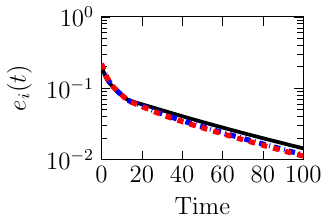}
}
\hfill
\subfloat[Modal truncation]{
\includegraphics[width=0.47\columnwidth]{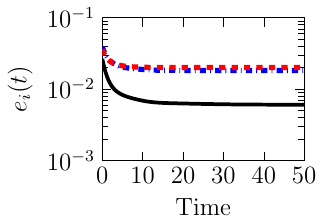}
}

\vspace{-1mm}

\caption{
Robustness of the proposed control approach.
(a) Evolution of the $L^2$ density error norm under three initial
perturbation amplitudes. The solid black, blue dash-dotted, and red
dashed curves correspond to perturbation amplitudes $10^{-1}$,
$5\times10^{-1}$, and $10^{0}$, respectively.
(b) Evolution of the $L^2$ density error norm in the presence of
higher-order modal perturbations. The solid black curve denotes the
nominal case, while the blue dash-dotted and red dashed curves
correspond to perturbations involving the mode sets $\{6,7,8\}$ and
$\{7,8,9\}$, respectively.
}
\label{fig:robustness_results}
\end{figure}

Overall, Fig.~\ref{fig:robustness_results} shows that the proposed
feedback controller may remain effective for the large initial
perturbations and may retain bounded closed-loop density errors in the
presence of higher-order Fourier mode perturbations neglected during controller synthesis. These results
provide numerical evidence of the robustness of the
interaction modulation framework beyond the nominal conditions used
for controller synthesis.
\subsection{Microscopic validation with finite agent populations}

We next evaluate the performance of the proposed control framework with
microscopic multi-agent models, to assess the effect of varying population sizes in relation to the considered continuification approach. We use the
same desired density, equilibrium interaction kernel, and feedback
controller as in Section~\ref{sec:initial_perturbations}.
Simulations are performed for agent populations of
$20,50,100$, and $200$. For each population, the desired density is
discretized into agent positions, and independent zero-mean Gaussian
perturbations with standard deviation $10^{-2}$ are added to generate
the initial configuration. The agent positions are then evolved directly
according to the microscopic dynamics in \eqref{eq:micro}. At each controller evaluation, the density is reconstructed from the
current agent positions using a periodic Gaussian kernel density
estimator with bandwidth $\sigma=0.2$. Its Fourier coefficients
are used to evaluate the feedback law in~\eqref{eq:state_feedback}, and the resulting
modification of the interaction kernel is applied directly to the
microscopic agent dynamics in \eqref{eq:micro}. This process is repeated at each time step.

For performance evaluation, the reconstructed density is evaluated
on a uniform spatial grid of $N_x=800$ points, and its $L^2$ error
relative to the desired density is normalized by its initial value.
Figure~\ref{fig:discrete_validation} shows that the density error
decreases for all considered populations and approaches a
population-dependent residual level. The residual error decreases as
the number of agents increases, reflecting the improved approximation
of the continuum density approximation with larger agent populations.
\begin{figure}[!t]
\centering
\includegraphics[width=0.47\columnwidth]{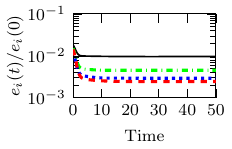}
\caption{Microscopic multi-agent validation for different agent
populations. The curves show the normalized $L^2$ error between the
KDE-reconstructed and desired densities. The solid black, green
dash-dotted, blue dotted, and red dashed curves correspond to
$20$, $50$, $100$, and $200$ agents, respectively.}
\label{fig:discrete_validation}
\end{figure}
These results demonstrate that the feedback controller designed from
the reduced mean-field model remains effective when applied directly
to the microscopic multi-agent dynamics, provided that enough agents are present.
\subsection{Numerical investigation of controllability}

We complement the analytical result of
Theorem~\ref{thm:controllability} with a numerical investigation of
controllability under broader conditions than those in~\eqref{eq:support_conditions}. Specifically, we examine how
the distribution of the retained Fourier modes between the equilibrium
density and the equilibrium kernel affects controllability of the
reduced linearized system.
For this study, we consider $N=10$ and assign each conjugate pair of retained non-zero Fourier modes either to the equilibrium density or to the equilibrium
kernel, consistently with the disjoint-support property implied by
\eqref{eq:modal_equilibrium}. Conditions where both coefficients are zero are neglected, as we expect these conditions to offer worse controllability based on~\eqref{eq:factorized_matrices}. All $2^{10}$
possible mode assignments are considered. For each assignment,
$100$ random realizations of the corresponding non-zero Fourier
coefficients are generated. For each realization, the matrices
$\mathbf A$ and $\mathbf B$ are constructed and the rank of the
controllability matrix in~\eqref{eq:controllability_matrix} is
evaluated. 
Figure~\ref{fig:kernel_mode_success} reports the percentage of controllable equilibria as the
fraction of active interaction kernel modes increases.
When only a small percentage of modes is assigned to the interaction
kernel, the equilibrium density retains a larger fraction of active
Fourier modes and, consequently, richer modal coupling through
$\mathbf B$. Nearly all realizations are controllable in this regime, supporting and extending our conclusions from Theorem~\ref{thm:controllability}.
As additional modes are assigned to the equilibrium kernel, fewer
active density modes remain available in $\mathbf B$, and the
percentage of controllable realizations decreases, as the modes are not achievable through the dynamics via powers of $\mathbf{A}$. In the numerical
experiments, controllable equilibria remain close to $100\%$
for up to $40\%$ active kernel modes, decrease sharply around $50\%$ active
modes, and becomes nearly zero with over $60\%$ active modes. The standard deviation increases near this transition, indicating
greater variability in controllability close to the transition.

\begin{figure}[!t]
\centering
\includegraphics[width=0.47\columnwidth]{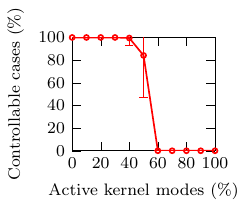}
\caption{Percentage of controllable realizations versus percentage of
active kernel modes ($N=10$). Each point represents the average percentage of controllable realizations over
all mode assignments with the corresponding percentage of active kernel
modes. Error bars indicate one standard deviation across realizations. For each assignment, 100 random realizations are considered.}
\label{fig:kernel_mode_success}
\end{figure}

These results extend the numerical investigation beyond the
conditions considered in
Theorem~\ref{thm:controllability}. They indicate that controllability
is strongly influenced by the allocation of Fourier modes between
equilibrium density and interaction kernel. In particular,
increasing the fraction of active kernel modes reduces the 
support available for the equilibrium density in $\mathbf B$ and can increase its rank
deficiency, requiring the dynamics generated by $\mathbf A$ to recover
additional missing columns. We find that $\mathbf{A}$ cannot recover such a rank deficiency above a certain percentage of inactive density modes.
\section{Conclusion}
\label{sec:results}

This paper introduces a new control approach framework to regulate collective behavior in large-scale,
nonlinear MASs via modulation of the interactions between agents. Motivated by biological
collective systems in which local changes in interaction rules
generate coordinated global behaviors, the proposed
framework treats the interaction law between agents as the control variable,
rather than applying direct control actions to individual agents.
Grounded in a continuified formulation that targets density control of large-scale MASs on a ring, we adopt a Fourier representation of the density dynamics. Under reasonable hypotheses that account for the physical and control features of MASs, we introduce a finite-dimensional approximation of the dynamics. Next, we establish conditions on the interactions for a desired density profile to be an equilibrium; linearize the system around that equilibrium, obtaining interaction modulation as an input; and synthesize a state-feedback controller to retrieve such an input, based on desired properties of the response. Numerical simulations on both microscopic and density dynamics demonstrate that the proposed controller can drive agents' density to a desired distribution, even under large initial perturbations and
presence of neglected Fourier modes. These results highlight the feasibility of our approach for designing interactions between agents to control and stabilize their arrangement.

As we introduce a new approach to density control, questions about controllability naturally arise. Focusing on application-relevant conditions, we establish controllability requirements for the finite-dimensional system approximation, which we expect to be verified in the vast majority of cases of interest for our framework. We corroborate and complement our controllability proof with numerical simulations on controllability of random systems, studying conditions beyond the hypotheses of our theorem. We find that our control framework is mostly feasible for desired densities with non-sparse spectral content, which represent the biggest share of densities targeted in practical control applications.

Our control framework comes with two main limitations. First, while our Fourier approach can be easily extended to two- and three-dimensional periodic domains, it does not allow to achieve non-periodic densities. In future work, we will seek to translate it to non-periodic domains. A promising approach involves the use of a different orthogonal complete functional basis that automatically satisfies Dirichlet boundary conditions of zero density at the domain borders, similar to assumed mode methods in Galerkin approximations\cite{Langtangen2019}. Second, the proposed approach is centralized. We plan to study decentralization approaches based on distributed estimation\cite{DiLorenzo2025} and leader-follower frameworks\cite{LeaderFollower,Plasticity} that leverage asymmetric interactions, keeping followers interactions fixed. Even with these limitations, our work provides the foundations for an innovative and promising approach to density control in MASs.
\section*{Acknowledgements}
We thank Petras Swissler and Sachit Butail for useful discussion about the problem.

\bibliographystyle{IEEEtran}
\bibliography{references}
\begin{IEEEbiography}
[{\includegraphics[width=1in,height=1.5in,clip,keepaspectratio]
{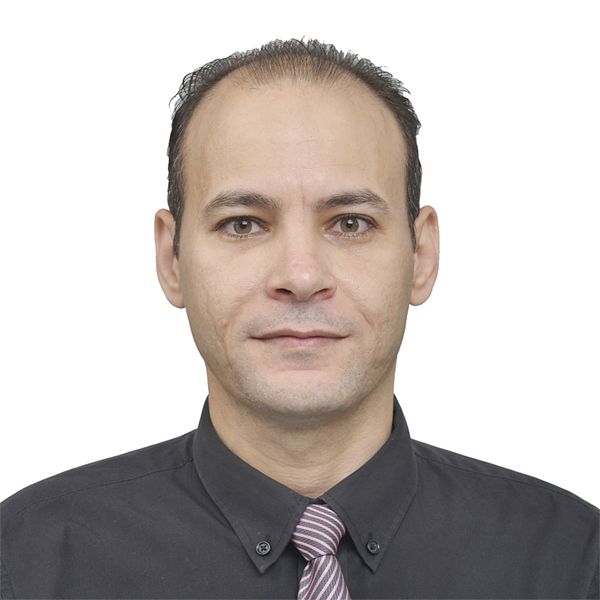}}]
{Mena Zkhary}
received the B.Sc. degree in Industrial Electronics and Control
Engineering and the M.Sc. degree in Engineering Mathematics from
the Faculty of Electronic Engineering, Menoufia University, Egypt,
in 2013 and 2019, respectively. From 2020 to 2025, he was a
Researcher and Assistant Lecturer of Mathematics at Ghent University Global Campus,
South Korea. He is currently pursuing the Ph.D. degree in Mechanical
Engineering at the New York Institute of Technology. His research
interests include modeling and control of dynamical systems,
large-scale multi-agent systems, and collective dynamics.
\end{IEEEbiography}

\begin{IEEEbiography}
[{\includegraphics[width=1in,height=1.5in,clip,keepaspectratio]
{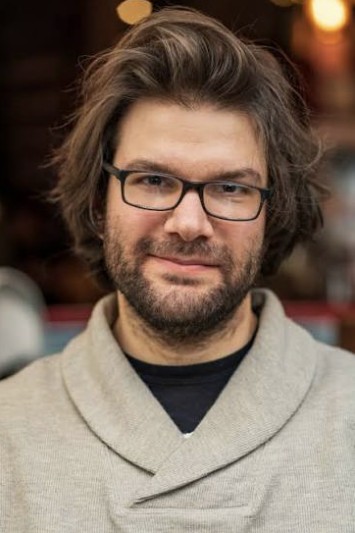}}]
{Alain Boldini}
is an Assistant Professor of Mechanical Engineering at the
New York Institute of Technology, since 2023. He received his B.Sc. in Aerospace Engineering from Politecnico di Milano in 2015, his M.Sc. in
Aerospace Engineering from Politecnico di Milano and Torino
in 2017, and his Ph.D. in Mechanical Engineering from
New York University in 2022. He has been a Postdoctoral Associate at New York University Tandon School of Engineering from 2022 to 2023. His research focuses on complex
systems modeling and control.
\end{IEEEbiography}
\end{document}